\documentclass[aps,prd,twocolumn,superscriptaddress,nofootinbib,10pt]{revtex4-2}

\usepackage{amsmath,amssymb,mathrsfs,bm,xcolor}
\usepackage{amsthm}
\usepackage[colorlinks=true,linkcolor=blue,citecolor=blue,urlcolor=blue]{hyperref}
\usepackage{orcidlink}

\newtheorem{proposition}{Proposition}
\newtheorem{theorem}{Theorem}
\newtheorem{corollary}{Corollary}
\newtheorem{remark}{Remark}

\newcommand{\Pln}{\mathcal{P}_{\ell n}}
\newcommand{\Ccross}{\mathcal{C}}
\newcommand{\Lie}{\mathcal{L}}

\begin{document}

\title{Horizon-regular cross-focusing and inner-horizon obstructions in spherical $f(R)$ gravity}

\author{Maickol Muñoz-Palma\orcidlink{0009-0008-1017-4406}}
\email{m.munoz58@ufromail.cl}
\affiliation{Departamento de Ciencias Físicas, Universidad de La Frontera, Casilla 54-D, 4811186 Temuco, Chile.}

\author{Francisco S. N. Lobo\orcidlink{0000-0002-9388-8373
}}
\email{fslobo@ciencias.ulisboa.pt}
\affiliation{Instituto de Astrof\'isica e Ci\^encias do Espa\c{c}o, Faculdade de Ci\^encias da Universidade de Lisboa, Edif\'icio C8, Campo Grande, P-1749-016 Lisbon, Portugal}
\affiliation{Departamento de F\'isica, Faculdade de Ci\^encias da Universidade de Lisboa, Edif\'icio C8, Campo Grande, P-1749-016 Lisbon, Portugal}

\author{Jean B\'aez Cuevas\orcidlink{0000-0002-3308-3362}}
\email{jean.baez@pucv.cl}
\affiliation{
Instituto de F\'isica, Pontificia Universidad Cat\'olica de Valpara\'iso, Casilla 4950, Valpara\'iso, Chile.}

\author{Francisco Tello-Ortiz \orcidlink{0000-0002-7104-5746}}
\email{francisco.tello@ufrontera.cl}
\affiliation{Departamento de Ciencias Físicas, Universidad de La Frontera, Casilla 54-D, 4811186 Temuco, Chile.}

\begin{abstract}
	We formulate a horizon-regular double-null criterion for regular inner
	marginal horizons in spherically symmetric metric $f(R)$ gravity.
	Using normalized outgoing and ingoing radial null vectors
	$\ell^\mu$ and $n^\mu$, with $n^\mu$ affinely parametrized, we
	derive an exact evolution law for the area-weighted outgoing expansion
	$r^2\theta_{(\ell)}$. Its source is controlled by the scalaron
	$F\equiv f_R>0$ and by a mixed quantity $\mathcal{P}_{\ell n}$
	containing matter, scalaron derivatives, and the curvature potential.
	If $\mathcal{P}_{\ell n}\leq F/r^2$ along a regular ingoing null
	segment issuing from a nondegenerate future outer marginal sphere,
	then the outgoing expansion cannot return to zero, and no second
	regular marginal sphere of the same family can occur on that
	generator. Conversely, a nondegenerate future inner marginal sphere
	requires the reverse inequality, so an outer--inner pair necessarily
	entails a source reversal and an exact integral balance. No
	trapped-region assumption is required.
	In the static limit, the criterion reduces to a horizon-regular
	relation involving the radial derivative of the metric function and
	remains valid in the degenerate case under the stated regularity
	conditions. It reproduces the Reissner--Nordstr\"om classification and
	is verified in an exact charged, nonconstant-curvature $f(R)$ black
	hole with a nonconstant scalaron. The resulting Cauchy-horizon
	statement is conditional and applies only when the candidate boundary
	is also a regular nondegenerate future inner marginal horizon.
\end{abstract}

\maketitle

\section{Introduction}
\label{sec:introduction}

Null focusing provides one of the most direct links between local
curvature and the causal structure of spacetime. The Raychaudhuri
equation governs the evolution of null congruences and lies at the
foundation of the classical focusing and singularity
theorems~\cite{Raychaudhuri:1953yv,Hawking:1970zqf,Wald1984}. Its
modified-gravity implications have also been explored, for instance in
inhomogeneous $f(R)$ cosmologies, where the additional curvature
degree of freedom modifies the usual matter-driven focusing
balance~\cite{Chakraborty:2023ork}. In spherical symmetry, the two radial
null expansions provide a natural quasi-local characterization of
trapped, marginal, and untrapped spheres, while their transverse
derivatives determine the future/past and outer/inner trapping-horizon
classification introduced by Hayward~\cite{Hayward:1993wb}. Cross-focusing
is therefore an established ingredient of the dual-null formalism, with
recent warped-spacetime analyses further clarifying its geometric
structure and relation to alternative spherical
decompositions~\cite{Maciel:2024vqg}.

The interior of a black hole presents a qualitatively different problem
from that of its outer horizon. In charged or rotating solutions of
General Relativity, an inner Killing horizon may also constitute a
Cauchy horizon, beyond which evolution from an initial Cauchy surface
ceases to be uniquely determined. Such boundaries are closely connected
with the mass-inflation instability~\cite{Poisson:1990eh,Ori:1991zz},
although their existence, regularity, and nonlinear stability are
logically distinct issues. Motivated by these questions, various
no-inner-horizon and no-Cauchy-horizon results have been established in
the presence of charged scalar hair, nonlinear electrodynamics,
Horndeski interactions, and Gauss--Bonnet
couplings~\cite{Cai:2020wrp,An:2021plu,Devecioglu:2021xug,Devecioglu:2023hmn}. More recently, a general no-inner-horizon criterion
has been obtained for a broad class of static anisotropic black holes in
Einstein gravity~\cite{Peng:2026baq}, while complementary
energy-condition arguments constrain the multiplicity of horizons in
static geometries~\cite{Yang:2021civ}.

Metric $f(R)$ gravity provides a particularly natural setting in which
to extend this programme. Its field equations introduce the scalaron
$F\equiv f_R$, whose derivatives enter the curvature projections
together with the matter sector~\cite{Sotiriou:2008rp,DeFelice:2010aj,Nojiri:2010wj,Capozziello:2011et}. Beyond its extensive applications to
cosmology and inflation~\cite{Sebastiani:2015kfa}, the theory
admits a rich black-hole sector containing reconstructed vacuum
solutions, nonconstant-curvature geometries, rotating configurations,
and charged matter-supported
solutions~\cite{Multamaki:2006zb,Sebastiani:2010kv,Karakasis:2021rpn,Chaturvedi:2023ctn,Hurtado:2020gic,Tang:2019qiy}.
Topological and thermodynamic properties of $R^2$ and $F(R)$ black
holes have likewise been studied~\cite{Cognola:2015wqa,EslamPanah:2023pce,Xu:2024yis}. Of particular relevance here, double-null and Kruskal-type
simulations indicate that the scalaron may undergo strong evolution in
black-hole interiors, approach limiting values, or develop singular
behaviour near inner-horizon or central
regions~\cite{Hwang:2011kg,Guo:2013dha,Guo:2015ira}.

A purely static-coordinate analysis is insufficient for the
inner-horizon problem~\cite{Munoz-Palma:2026dbd}. Between an outer event horizon and a
static inner Cauchy horizon, the usual temporal and radial coordinates
exchange causal character, while a zero of a static metric function
alone does not provide a global characterization of a Cauchy horizon.
A physically appropriate analysis should therefore be formulated in
terms of null quantities that remain regular across the outer horizon
and throughout the intervening region.

The purpose of this work is to formulate a horizon-regular criterion
for spherical metric $f(R)$ gravity. Starting from the standard
cross-focusing equation, we isolate the exact Jordan-frame
matter--scalaron source and recast the dynamics, in an affine null
gauge, as an area-weighted evolution law for
$r^2\theta_{(\ell)}$. This yields a propagation obstruction: if
$\mathcal{P}_{\ell n}\leq F/r^2$ along a regular ingoing null segment
from a nondegenerate future outer marginal sphere, the outgoing
expansion cannot return to zero. Conversely, a regular future inner
marginal sphere requires the reverse inequality locally, so an
outer--inner pair demands both a source reversal and an exact integral
balance. We further relate this criterion to the corresponding static
radial identity and test it in General-Relativity limits and an exact
charged nonconstant-curvature $f(R)$ black hole with $F'\neq0$.

The result should be interpreted as an obstruction to a specified class
of regular inner marginal horizons, rather than as an unconditional
global no-Cauchy-horizon theorem. In Reissner--Nordstr\"om-type static
extensions, a regular inner Killing horizon may simultaneously be
marginal and constitute a Cauchy horizon, in which case the criterion
applies under the stated hypotheses. In a general dynamical spacetime,
however, a Cauchy horizon need not coincide with an inner marginal or
trapping horizon, and singular mass-inflation boundaries require a
separate analysis.

The paper is organized as follows. Section~\ref{sec:field-equations}
introduces the $f(R)$ field equations and double-null geometry.
Section~\ref{sec:focusing} develops the affine and cross-focusing
relations, and Sec.~\ref{sec:theorems} establishes the integral
criterion, propagation obstruction, and necessary source reversal.
Section~\ref{sec:static-limit} discusses the regular static limit,
while Sec.~\ref{sec:checks} presents General-Relativity checks,
constant-curvature reductions, and a scalaron-active charged example.
The scope of the Cauchy-horizon statement and its relation to mass
inflation are discussed in Sec.~\ref{sec:scope}. We conclude in
Sec.~\ref{sec:conclusions}, and collect the coordinate derivation in
Appendix~\ref{app:derivation}.

\section{\texorpdfstring{Metric $f(R)$ gravity in double-null form}
	{Metric f(R) gravity in double-null form}}
\label{sec:field-equations}

\subsection{Field equations and scalaron trace equation}

We consider metric $f(R)$ gravity with action
\begin{equation}
	S=\frac{1}{16\pi}\int d^4x\,\sqrt{-g}\,f(R)
	+S_m[g_{\mu\nu},\Psi] .
	\label{eq:action}
\end{equation}
We adopt the metric signature $(-,+,+,+)$ and the curvature
convention
\begin{equation}
	R^\rho{}_{\sigma\mu\nu}
	=\partial_\mu\Gamma^\rho{}_{\nu\sigma}
	-\partial_\nu\Gamma^\rho{}_{\mu\sigma}
	+\Gamma^\rho{}_{\mu\lambda}
	\Gamma^\lambda{}_{\nu\sigma}
	-\Gamma^\rho{}_{\nu\lambda}
	\Gamma^\lambda{}_{\mu\sigma},
\end{equation}
with $R_{\mu\nu}=R^\rho{}_{\mu\rho\nu}$.
Variation with respect to $g_{\mu\nu}$ gives
\begin{equation}
	F R_{\mu\nu}-\frac12 f g_{\mu\nu}
	-\nabla_\mu\nabla_\nu F+g_{\mu\nu}\Box F
	=8\pi T_{\mu\nu},
	\label{eq:field-eq}
\end{equation}
where
$F\equiv f_R=df/dR$, $\Box\equiv g^{\alpha\beta}\nabla_\alpha\nabla_\beta$,
and
\begin{equation}
	T_{\mu\nu}
	\equiv-\frac{2}{\sqrt{-g}}
	\frac{\delta S_m}{\delta g^{\mu\nu}}
\end{equation}
is the matter energy-momentum tensor.
Taking the trace of Eq.~\eqref{eq:field-eq} yields
\begin{equation}
	3\Box F+FR-2f=8\pi T,
	\label{eq:trace}
\end{equation}
with $T\equiv T^\mu{}_\mu$.

Throughout the theorem-bearing part of the paper, we restrict the
analysis to spacetime regions on which the areal radius satisfies
$r>0$. We assume that the metric and the scalaron $F$ are at least
twice continuously differentiable, that the matter projections
entering the focusing equations are continuous, and that all
quantities appearing in those equations remain finite on the null
segments under consideration. Boundary values are understood as
one-sided limits from within the regular region.

We further assume $F>0$. This condition ensures a positive effective gravitational coupling,
prevents the coefficient of the Ricci tensor in
Eq.~\eqref{eq:field-eq} from vanishing, and permits division by $F$
in the focusing inequalities derived below. It is a standard
viability requirement in metric $f(R)$ gravity, although it does
not by itself constitute a complete set of stability
conditions~\cite{Sotiriou:2008rp,DeFelice:2010aj}.
The original undivided field equations, including
Eq.~\eqref{eq:Einstein-form} (see below), may remain algebraically meaningful at
a zero of $F$. However, the divided Einstein-like equations, the
source ratio $\Pln/F$, and the theorems presented below do not apply
across such a point.

It is convenient to rewrite Eq.~\eqref{eq:field-eq} in Einstein form,
\begin{equation}
	F G_{\mu\nu}=8\pi T_{\mu\nu}
	+\nabla_\mu\nabla_\nu F-g_{\mu\nu}\Box F
	+\frac12(f-FR)g_{\mu\nu}.
	\label{eq:Einstein-form}
\end{equation}
Equivalently, on a region satisfying $F>0$,
one may divide this equation by $F$ and interpret the scalaron
derivative and potential terms as an effective gravitational
source. In the present work, however, these contributions will be
kept explicit rather than absorbed into an effective
energy-momentum tensor.

\subsection{Double-null geometry and null frame}

In spherical symmetry, we adopt double-null coordinates $(u,v)$ and
write the metric as
\begin{equation}
	ds^2=-2e^{-2\sigma(u,v)}\,du\,dv
	+r^2(u,v)\,d\Omega^2.
	\label{eq:double-null}
\end{equation}
We assume that both $u$ and $v$ increase toward the future and that
the conformal factor $e^{-2\sigma}$ is finite and strictly positive
throughout the regular region under consideration.

Along each curve $v=\mathrm{const.}$, we choose the future-directed
ingoing null vector $n^\mu$ and the future-directed outgoing null
vector $\ell^\mu$ as
\begin{equation}
	n^\mu=e^{2\sigma}(\partial_u)^\mu,
	\qquad
	\ell^\mu=(\partial_v)^\mu.
	\label{eq:null-frame}
\end{equation}
They satisfy
\begin{equation}
	n^\mu n_\mu=0,
	\qquad
	\ell^\mu\ell_\mu=0,
	\qquad
	\ell^\mu n_\mu=-1.
\end{equation}

The vector $n^\mu$ is tangent to the ingoing null generators and is
affinely parametrized. Thus, if $\lambda$ denotes the affine
parameter increasing toward the black-hole interior, then
\begin{equation}
	n^\mu=\left(\frac{d}{d\lambda}\right)^\mu,
	\qquad
	\nabla_n n^\mu=0,
	\qquad
	\nabla_n\ell^\mu=0.
\end{equation}
The last relation fixes $\ell^\mu$ by parallel transport along each
ingoing generator.

The normalization and transport conditions remain invariant under
\begin{equation}
	n^\mu\rightarrow c\,n^\mu,
	\qquad
	\ell^\mu\rightarrow c^{-1}\ell^\mu,
	\qquad
	c>0,
	\qquad
	n(c)=0.
	\label{eq:boost}
\end{equation}
Here $n(c)\equiv n^\mu\nabla_\mu c=dc/d\lambda$. Thus,
$n(c)=0$ requires the positive boost function $c$ to remain
constant along each affinely parametrized ingoing null generator,
although it may vary from one generator to another. Under this
transformation, the null expansions scale according to
\begin{equation}
	\theta_{(n)}\rightarrow c\,\theta_{(n)},
	\qquad
	\theta_{(\ell)}\rightarrow c^{-1}\theta_{(\ell)}.
\end{equation}
Here
$T_{\ell n}\equiv T_{\mu\nu}\ell^\mu n^\nu$ and
$\nabla_\ell\nabla_n F
\equiv
\ell^\mu n^\nu\nabla_\mu\nabla_\nu F$.
Because $n(c)=0$, the quantities
$\Lie_n\theta_{(\ell)}$,
$\theta_{(\ell)}\theta_{(n)}$,
$T_{\ell n}$, and
$\nabla_\ell\nabla_n F$
are invariant. Consequently, the cross-focusing equation, the
outer/inner classification, and the associated source threshold are
independent of this residual boost freedom.

The null expansions of the symmetry spheres are
\begin{equation}
	\theta_{(\ell)}
	=\frac{2}{r}\,\ell(r)
	=\frac{2r_{,v}}{r},
	\qquad
	\theta_{(n)}
	=\frac{2}{r}\,n(r)
	=\frac{2e^{2\sigma}r_{,u}}{r}.
\end{equation}
A future marginal sphere of the outgoing family is characterized by $\theta_{(\ell)}=0$ and $\theta_{(n)}<0$, whereas a future trapped sphere satisfies $\theta_{(\ell)}<0$ and $\theta_{(n)}<0$.

\section{Affine focusing and cross-focusing}
\label{sec:focusing}

\subsection{Affine null Raychaudhuri equation}

Let $k^\mu$ be tangent to a future-directed affinely parametrized
radial null congruence, with affine parameter $\lambda$:
\begin{equation}
	k^\mu=\left(\frac{d}{d\lambda}\right)^\mu,
	\qquad
	k^\nu\nabla_\nu k^\mu=0.
	\label{eq:affine-k}
\end{equation}
We denote the expansion of the congruence by $\theta_{(k)}$ and define
$T_{kk}\equiv T_{\mu\nu}k^\mu k^\nu$.
Spherical radial congruences are hypersurface orthogonal and shear
free. The Raychaudhuri equation and the twice-null projection of the
field equations therefore give
\begin{align}
	\frac{d\theta_{(k)}}{d\lambda}
	&=-\frac12\theta_{(k)}^2-R_{\mu\nu}k^\mu k^\nu,
	\label{eq:Raychaudhuri-affine}\\
	F R_{\mu\nu}k^\mu k^\nu
	&=8\pi T_{kk}+\frac{d^2F}{d\lambda^2},
	\label{eq:null-projection}
\end{align}
where the second equality uses affine parametrization. Hence
\begin{equation}
	\frac{d\theta_{(k)}}{d\lambda}
	=-\frac12\theta_{(k)}^2
	-\frac{1}{F}\left(8\pi T_{kk}
	+\frac{d^2F}{d\lambda^2}\right).
	\label{eq:modified-Raychaudhuri}
\end{equation}
Thus the matter null energy condition alone does not fix
same-direction focusing: the scalaron acceleration may reinforce or
oppose the matter term. This relation is a complementary diagnostic;
the existence of an inner marginal horizon is controlled by the
transverse evolution of $\theta_{(\ell)}$, to which we now turn.

\subsection{Exact cross-focusing equation}

For the frame~\eqref{eq:null-frame}, we define
\begin{equation}
	G_{\ell n}\equiv G_{\mu\nu}\ell^\mu n^\nu,
	\qquad
	\Lie_n\theta_{(\ell)}
	\equiv n^\mu\nabla_\mu\theta_{(\ell)}.
\end{equation}
The mixed Einstein tensor obeys the geometric identity
\begin{equation}
	G_{\ell n}
	=\Lie_n\theta_{(\ell)}
	+\theta_{(\ell)}\theta_{(n)}
	+\frac{1}{r^2},
	\label{eq:Gln-geometry}
\end{equation}
Equation~\eqref{eq:Gln-geometry} is the spherical specialization of
the standard cross-focusing relation for a normalized null pair
~\cite{Hayward:1993wb,Maciel:2024vqg}. The novelty below lies not in this
geometric identity itself, but in its explicit Jordan-frame $f(R)$
source decomposition and in the area-weighted propagation theorem.
A direct coordinate derivation in the conventions adopted here is
provided in Appendix~\ref{app:derivation}.

We define the mixed matter--scalaron source
\begin{equation}
	\Pln\equiv
	8\pi T_{\ell n}
	+\nabla_\ell\nabla_n F
	+\Box F
	+\frac12(FR-f).
	\label{eq:P-invariant}
\end{equation}
Here and throughout, $\nabla_\ell\nabla_n F$ denotes the tensorial
contraction of the Hessian defined above. It should not, in general, be
identified with the iterated directional derivative $\ell[n(F)]$,
since the latter also contains a derivative of $n^\mu$.

{In General Relativity, with
	$f(R)=R-2\Lambda$ and hence $F\equiv1$, the scalaron derivative
	terms in Eq.~\eqref{eq:P-invariant} vanish, while the algebraic
	curvature term reduces to the cosmological constant. Consequently,
	$\Pln=8\pi T_{\ell n}+\Lambda$ (Sec.~\ref{sec:checks}), so the mixed
	source is determined entirely by the physical matter content and
	$\Lambda$. In metric $f(R)$ gravity, by contrast, the mixed source
	contains additional scalaron-dependent contributions through the
	mixed Hessian $\nabla_\ell\nabla_n F$, the trace term $\Box F$, and
	the generally nontrivial algebraic curvature potential
	$\tfrac12(FR-f)$. These are distinct contributions, although they are
	not generically independent because they are related by the scalaron
	trace equation. They may reinforce or oppose the matter contribution
	and can participate in the source reversal required for a regular
	inner marginal horizon, as illustrated explicitly in
	Sec.~\ref{sec:Tang-benchmark}. This additional scalaron structure is
	the essential modification of the cross-focusing source when passing
	from General Relativity to metric $f(R)$ gravity.}

Contracting Eq.~\eqref{eq:Einstein-form} with
$\ell^\mu n^\nu$ and using $g_{\ell n}=-1$ gives
\begin{equation}
	F G_{\ell n}=\Pln,
	\label{eq:P-def-eq}
\end{equation}

Combining Eqs.~\eqref{eq:Gln-geometry}
and~\eqref{eq:P-def-eq}, we obtain the exact cross-focusing equation
\begin{equation}
	\Lie_n\theta_{(\ell)}
	=-\theta_{(\ell)}\theta_{(n)}
	-\frac{1}{r^2}
	+\frac{\Pln}{F}.
	\label{eq:cross-focusing}
\end{equation}
This relation holds on every regular region satisfying $F>0$. The undivided relation
$F G_{\ell n}=\Pln$ remains algebraically meaningful without
division by $F$, but Eq.~\eqref{eq:cross-focusing} and the
inequalities derived from it do not apply through a zero of the
scalaron.

A particularly useful consequence follows by multiplying
Eq.~\eqref{eq:cross-focusing} by $r^2$. Since
\begin{equation}
	\Lie_n(r^2)
	=2r\,n(r)
	=r^2\theta_{(n)},
	\label{eq:area-derivative}
\end{equation}
the product term involving the two expansions combines into a total
derivative:
\begin{align}
	\Lie_n\!\left(r^2\theta_{(\ell)}\right)
	&=r^2\Lie_n\theta_{(\ell)}
	+\theta_{(\ell)}\Lie_n(r^2)
	\notag\\
	&=-1+\frac{r^2\Pln}{F}.
	\label{eq:area-weighted-cross-focusing}
\end{align}
Equation~\eqref{eq:area-weighted-cross-focusing} is the
area-weighted cross-focusing identity. It isolates the geometric
threshold $\Pln=F/r^2$ directly and provides the natural
monotonicity law for the analysis of inner marginal horizons.

{The quantity $F/r^2$ combines the scalaron
	coupling function $F$, which multiplies the Ricci tensor in the
	metric $f(R)$ field equations, with $1/r^2$, the intrinsic Gauss
	curvature of the symmetry sphere. It therefore defines a
	scalaron-weighted local spherical-curvature scale, rather than an
	independent input to the theory. By contrast, $\Pln$, as defined in
	Eq.~\eqref{eq:P-invariant}, is the mixed Jordan-frame
	matter--scalaron source that drives the transverse focusing of the
	outgoing congruence. The criterion developed below is therefore a
	comparison between this mixed source and the corresponding
	scalaron-weighted curvature threshold: whether $\Pln$ remains below
	or rises above the value $F/r^2$ required for the area-weighted
	outgoing expansion to reverse its monotonic trend and recover toward
	zero.}

In the coordinates of Eq.~\eqref{eq:double-null}, the source can be
written as
\begin{align}
	\Pln={}&8\pi e^{2\sigma}T_{uv}
	-e^{2\sigma}F_{,uv}
	\notag\\
	&-\frac{2e^{2\sigma}}{r}
	\left(r_{,u}F_{,v}+r_{,v}F_{,u}\right)
	+\frac12(FR-f).
	\label{eq:P-coordinate}
\end{align}
This expression follows from
\begin{equation}
	\nabla_\ell\nabla_n F=e^{2\sigma}F_{,uv}
	\label{eq:mixed-Hessian-coordinate}
\end{equation}
and
\begin{equation}
	\Box F=-2e^{2\sigma}
	\left[
	F_{,uv}
	+\frac{r_{,u}F_{,v}+r_{,v}F_{,u}}{r}
	\right].
	\label{eq:boxF-coordinate}
\end{equation}

Alternatively, the trace equation~\eqref{eq:trace} yields
\begin{equation}
	\Pln=8\pi T_{\ell n}
	+\nabla_\ell\nabla_n F
	+\frac{8\pi}{3}T
	+\frac16(FR+f).
	\label{eq:P-trace}
\end{equation}
Equations~\eqref{eq:P-invariant}, \eqref{eq:P-coordinate},
and~\eqref{eq:P-trace} are equivalent. Different forms may be more
useful depending on whether the metric, the scalaron profile, or the
matter trace is known explicitly.

To display the physical content of the mixed matter projection,
introduce an orthonormal radial frame $(u^\mu,s^\mu)$ adapted to
the normalized null pair,
\begin{equation}
	\ell^\mu=\frac{1}{\sqrt{2}}
	\left(u^\mu+s^\mu\right),
	\qquad
	n^\mu=\frac{1}{\sqrt{2}}
	\left(u^\mu-s^\mu\right),
	\label{eq:null-orthonormal-relation}
\end{equation}
where $u^\mu u_\mu=-1$, $s^\mu s_\mu=1$, and
$u^\mu s_\mu=0$. For an anisotropic matter source with energy
density $\rho$, radial pressure $p_r$, and no radial energy flux,
\begin{equation}
	T_{\ell n}=\frac12(\rho-p_r).
	\label{eq:matter-mixed}
\end{equation}
The two principal same-direction projections are
\begin{equation}
	T_{\ell\ell}=T_{nn}
	=\frac12(\rho+p_r).
	\label{eq:matter-same-direction}
\end{equation}
This distinction is important. A radial Maxwell field satisfies
$p_r=-\rho$, and therefore
\begin{equation}
	T_{\ell\ell}=T_{nn}=0,
	\qquad
	T_{\ell n}=\rho>0.
\end{equation}
It can thus saturate the same-direction null projection while
providing the positive mixed stress required to reverse the
cross-focusing balance at an inner horizon.

\subsection{Outer, inner, and degenerate marginal spheres}

At a regular future marginal sphere $H$ of the outgoing family, we
assume
\begin{equation}
	\theta_{(\ell)}\big|_H=0,
	\qquad
	\theta_{(n)}\big|_H<0,
	\qquad
	r_H>0,
	\qquad
	F_H>0,
\end{equation}
with all quantities entering the cross-focusing equation finite.
Equation~\eqref{eq:cross-focusing} then reduces to
\begin{equation}
	\Ccross_H
	\equiv
	\left.\Lie_n\theta_{(\ell)}\right|_H
	=-\frac{1}{r_H^2}
	+\frac{\left.\Pln\right|_H}{F_H}.
	\label{eq:C-horizon}
\end{equation}

Following the standard trapping-horizon classification
of Hayward~\cite{Hayward:1993wb}, a future marginal sphere is classified
according to
\begin{align}
	\text{future outer:}\quad
	&\Ccross_H<0,
	\label{eq:outer-def}\\
	\text{future inner:}\quad
	&\Ccross_H>0,
	\label{eq:inner-def}\\
	\text{degenerate:}\quad
	&\Ccross_H=0.
	\label{eq:degenerate-def}
\end{align}
Strictly speaking, these conditions classify an individual marginal
sphere. A hypersurface foliated by such spheres is called a future
outer, future inner, or degenerate trapping horizon when the
corresponding sign is maintained along the foliation. A marginal
sphere is nondegenerate whenever $\Ccross_H\neq0$.

The quantity $\Ccross_H$ is invariant under the residual boost
freedom
\begin{equation}
	n^\mu\rightarrow c\,n^\mu,
	\qquad
	\ell^\mu\rightarrow c^{-1}\ell^\mu,
	\qquad
	n(c)=0,
	\qquad
	c>0,
\end{equation}
introduced in Eq.~\eqref{eq:boost}. Indeed,
$\theta_{(n)}\rightarrow c\,\theta_{(n)}$ and
$\theta_{(\ell)}\rightarrow c^{-1}\theta_{(\ell)}$, while
$\Lie_n\theta_{(\ell)}$ remains invariant. The classification is therefore independent of the admissible
normalization of the null frame.

Equation~\eqref{eq:C-horizon} immediately gives
\begin{align}
	\left.\Pln\right|_H
	&<\frac{F_H}{r_H^2}
	&&\text{for a future outer marginal sphere},
	\label{eq:outer-threshold}\\
	\left.\Pln\right|_H
	&>\frac{F_H}{r_H^2}
	&&\text{for a future inner marginal sphere},
	\label{eq:inner-threshold}\\
	\left.\Pln\right|_H
	&=\frac{F_H}{r_H^2}
	&&\text{for a degenerate marginal sphere}.
\end{align}

Thus, a regular nondegenerate future inner marginal sphere requires
the mixed matter--scalaron source to exceed the scalaron-weighted
spherical-curvature threshold $F/r^2$. The inequalities above are
the $f(R)$ field-equation representation of the local Hayward
outer/inner classification. By themselves, they do not constitute a
global obstruction. The nontrivial propagation result derived below
is that maintaining the opposite inequality throughout the
intervening region prevents the outgoing expansion from returning to
zero.

\section{Obstructions to a regular inner marginal horizon}
\label{sec:theorems}

We now formulate the main results. Let $\gamma$ be a future-directed,
affinely parametrized ingoing radial null geodesic with tangent
$n^\mu$, and let $\lambda$ denote its affine parameter, increasing
from an outer marginal sphere toward the black-hole interior. We
assume throughout that $r>0$, $F>0$, and that all quantities
entering the cross-focusing equation remain regular on the segment
under consideration. If an endpoint is reached only as a limiting
null boundary, all endpoint values and integrals below are understood
in the corresponding one-sided sense; whenever an integral
is improper, its convergence is included among the hypotheses.

\begin{proposition}[Exact area-weighted cross-focusing integral]
	\label{prop:integral}
	Let $\lambda_a<\lambda_b$ be two points on $\gamma$. Then
	Eq.~\eqref{eq:area-weighted-cross-focusing} implies
	\begin{equation}
		\begin{split}
			&r^2(\lambda_b)\theta_{(\ell)}(\lambda_b)
			-r^2(\lambda_a)\theta_{(\ell)}(\lambda_a)
			\\
			&\qquad =
			\int_{\lambda_a}^{\lambda_b}
			\left(
			\frac{r^2\Pln}{F}-1
			\right)d\lambda .
		\end{split}
		\label{eq:area-weighted-integral}
	\end{equation}
	In particular, if the two endpoints are marginal spheres of the
	outgoing family,
	\begin{equation}
		\theta_{(\ell)}(\lambda_o)
		=\theta_{(\ell)}(\lambda_i)=0,
	\end{equation}
	then
	\begin{equation}
		\int_{\lambda_o}^{\lambda_i}
		\left(
		\frac{r^2\Pln}{F}-1
		\right)d\lambda=0.
		\label{eq:integral-criterion}
	\end{equation}
\end{proposition}

\begin{proof}
	Integrating
	Eq.~\eqref{eq:area-weighted-cross-focusing} along the affine
	generator gives Eq.~\eqref{eq:area-weighted-integral}. If both
	endpoints are marginal, the two boundary terms vanish, yielding
	Eq.~\eqref{eq:integral-criterion}.
\end{proof}

Unlike an integral derived in a Schwarzschild-like static chart,
Eq.~\eqref{eq:integral-criterion} contains no singular horizon weight.
It is formulated entirely in terms of a null frame that remains
regular across the outer marginal surface and through the nonstatic
interior. Moreover, the product
$\theta_{(\ell)}\theta_{(n)}$ has been absorbed exactly into the
area-weighted derivative.

\begin{theorem}[Sufficient obstruction to a second regular future marginal sphere]
	\label{thm:no-inner}
	Let $\lambda_o$ correspond to a regular, nondegenerate future
	outer marginal sphere of the outgoing family, so that
	\begin{equation}
		\theta_{(\ell)}(\lambda_o)=0,
		\qquad
		\theta_{(n)}(\lambda_o)<0,
		\qquad
		\left.\Lie_n\theta_{(\ell)}\right|_{\lambda_o}<0.
		\label{eq:outer-initial-data}
	\end{equation}
	Suppose that, along the subsequent regular segment of $\gamma$,
	\begin{equation}
		\Pln\leq\frac{F}{r^2}.
		\label{eq:source-bound}
	\end{equation}
	Then $r^2\theta_{(\ell)}$ is nonincreasing along the segment and
	is strictly negative for every $\lambda>\lambda_o$. Consequently,
	$\theta_{(\ell)}$ cannot return to zero, and no later regular
	future marginal sphere of the outgoing family can occur on the same
	generator. In particular, the segment cannot terminate at a regular
	future inner marginal sphere.
\end{theorem}

\begin{proof}
	At the initial future outer marginal sphere,
	$\theta_{(\ell)}=0$, and therefore
	\begin{equation}
		\left.
		\Lie_n\!\left(r^2\theta_{(\ell)}\right)
		\right|_{\lambda_o}
		=
		r_o^2
		\left.\Lie_n\theta_{(\ell)}\right|_{\lambda_o}
		<0.
	\end{equation}
	By continuity, $r^2\theta_{(\ell)}$ becomes strictly negative
	immediately to the future of $\lambda_o$. Furthermore,
	Eq.~\eqref{eq:source-bound} and the area-weighted identity give
	\begin{equation}
		\Lie_n\!\left(r^2\theta_{(\ell)}\right)
		=-1+\frac{r^2\Pln}{F}
		\leq0.
	\end{equation}
	Hence $r^2\theta_{(\ell)}$ cannot increase from its negative
	value back to zero. Since $r>0$, it follows that
	$\theta_{(\ell)}<0$ at every later regular point of the segment.
\end{proof}

Theorem~\ref{thm:no-inner} is stronger than a proof based directly on
the sign of
$-\theta_{(\ell)}\theta_{(n)}$: the exclusion of a second marginal
sphere does not require $\theta_{(n)}<0$ at every intermediate
point. In the physically relevant case of a future trapped region,
\begin{equation}
	\theta_{(\ell)}<0,
	\qquad
	\theta_{(n)}<0,
	\label{eq:trapped-assumption}
\end{equation}
the original cross-focusing equation additionally gives
\begin{equation}
	\Lie_n\theta_{(\ell)}
	\leq-\theta_{(\ell)}\theta_{(n)}<0,
\end{equation}
so that $\theta_{(\ell)}$ itself is then strictly decreasing.

The condition~\eqref{eq:source-bound} is sufficient rather than
necessary. Its value is that it separates the universal spherical
term from the mixed matter--scalaron source and can be tested locally
along a null evolution without first determining the complete global
causal structure.

The resulting monotonic obstruction, and its contrast with a possible source-driven return of the outgoing expansion to zero, are illustrated schematically in Fig.~\ref{fig:obstruction-schematic}.

\begin{corollary}
	[Conditional obstruction to a regular inner Killing/Cauchy horizon]
	\label{cor:Cauchy}
	Consider a static, spherically symmetric black hole admitting a
	regular double-null extension from a nondegenerate future outer
	Killing horizon through the inter-horizon region. Suppose that a
	candidate inner boundary is a regular, nondegenerate future inner
	Killing horizon, is foliated by marginal symmetry spheres of the
	outgoing family, and constitutes a Cauchy horizon in the maximal
	development under consideration. If $F>0$, the required regularity
	conditions, and Eq.~\eqref{eq:source-bound} hold along the ingoing
	null generators connecting the two horizons, then no such regular
	inner Killing/Cauchy horizon exists.
\end{corollary}

\begin{proof}
	Under the stated assumptions, each relevant generator begins at a
	nondegenerate future outer marginal sphere and reaches the candidate
	inner boundary as a regular future inner marginal sphere.
	Theorem~\ref{thm:no-inner} excludes such a second marginal boundary.
	The additional statement that this boundary is a Cauchy horizon is
	a global assumption and is not inferred from marginality alone.
\end{proof}

The corollary is intentionally conditional. It does not identify every
inner Killing horizon with a Cauchy horizon, nor does it imply that
every Cauchy horizon in a dynamical spacetime must be foliated by
marginal spheres.

\begin{proposition}[Necessary source reversal]
	\label{prop:reversal}
	Under the regularity and positivity assumptions above, any
	nondegenerate future inner marginal horizon must satisfy
	\begin{equation}
		\begin{split}
			8\pi T_{\ell n}
			+\nabla_\ell\nabla_n F
			+\Box F
			+\frac12(FR-f)
			>\frac{F}{r^2}
		\end{split}
		\label{eq:necessary-inner}
	\end{equation}
	at the horizon. Equivalently,
	\begin{equation}
		\begin{split}
			8\pi T_{\ell n}
			+\nabla_\ell\nabla_n F
			+\frac{8\pi}{3}T
			+\frac16(FR+f)
			>\frac{F}{r^2}.
		\end{split}
		\label{eq:necessary-inner-trace}
	\end{equation}
	If such an inner marginal horizon is connected to a nondegenerate
	future outer marginal sphere by a regular ingoing generator, then
	the function
	\begin{equation}
		\mathscr{S}(\lambda)
		\equiv\frac{r^2\Pln}{F}-1,
		\label{eq:source-reversal-function}
	\end{equation}
{which we will refer to as the source-reversal
	function, since its sign along $\gamma$ directly determines whether
	the area-weighted outgoing expansion is driven monotonically away
	from, or can evolve back toward, a second marginal sphere,}
	is negative at the outer marginal sphere and positive at the inner
	marginal sphere. If $\mathscr{S}$ is continuous, it must therefore
	vanish at least once between them.
\end{proposition}

\begin{proof}
	At a future inner marginal sphere,
	$\Ccross_H>0$. Equation~\eqref{eq:C-horizon} therefore gives
	\begin{equation}
		\left.\Pln\right|_H>\frac{F_H}{r_H^2}.
	\end{equation}
	Substitution of either Eq.~\eqref{eq:P-invariant} or
	Eq.~\eqref{eq:P-trace} yields
	Eqs.~\eqref{eq:necessary-inner}
	and~\eqref{eq:necessary-inner-trace}. At a nondegenerate future
	outer marginal sphere, the corresponding inequality is strict in
	the opposite direction. Continuity then implies that
	$\mathscr{S}$ crosses zero somewhere between the two marginal
	surfaces.
\end{proof}

Moreover, if an outer and an inner marginal sphere lie on the same
regular generator, Proposition~\ref{prop:integral} imposes the exact
balance
\begin{equation}
	\int_{\lambda_o}^{\lambda_i}
	\mathscr{S}(\lambda)\,d\lambda=0.
	\label{eq:source-balance}
\end{equation}
Thus, positive contributions above the threshold must compensate
exactly for the region in which the source lies below it.

Figure~\ref{fig:obstruction-schematic} summarizes the two complementary
outcomes. Panel~(a) shows that when
$\mathscr{S}(\lambda)\leq0$, equivalently
$\Pln\leq F/r^2$, the area-weighted expansion
$r^2\theta_{(\ell)}$ cannot recover from its negative value and
return to zero. By contrast, a possible second marginal sphere
requires the source to reverse the monotonic trend. Panel~(b) displays
the corresponding necessary behaviour of $\mathscr{S}$: for a
regular outer--inner pair it must change sign, while its negative and
positive contributions must satisfy the exact balance
\eqref{eq:source-balance}.

\begin{figure*}[th!]
	\centering
	\includegraphics[width=0.85\textwidth]
	{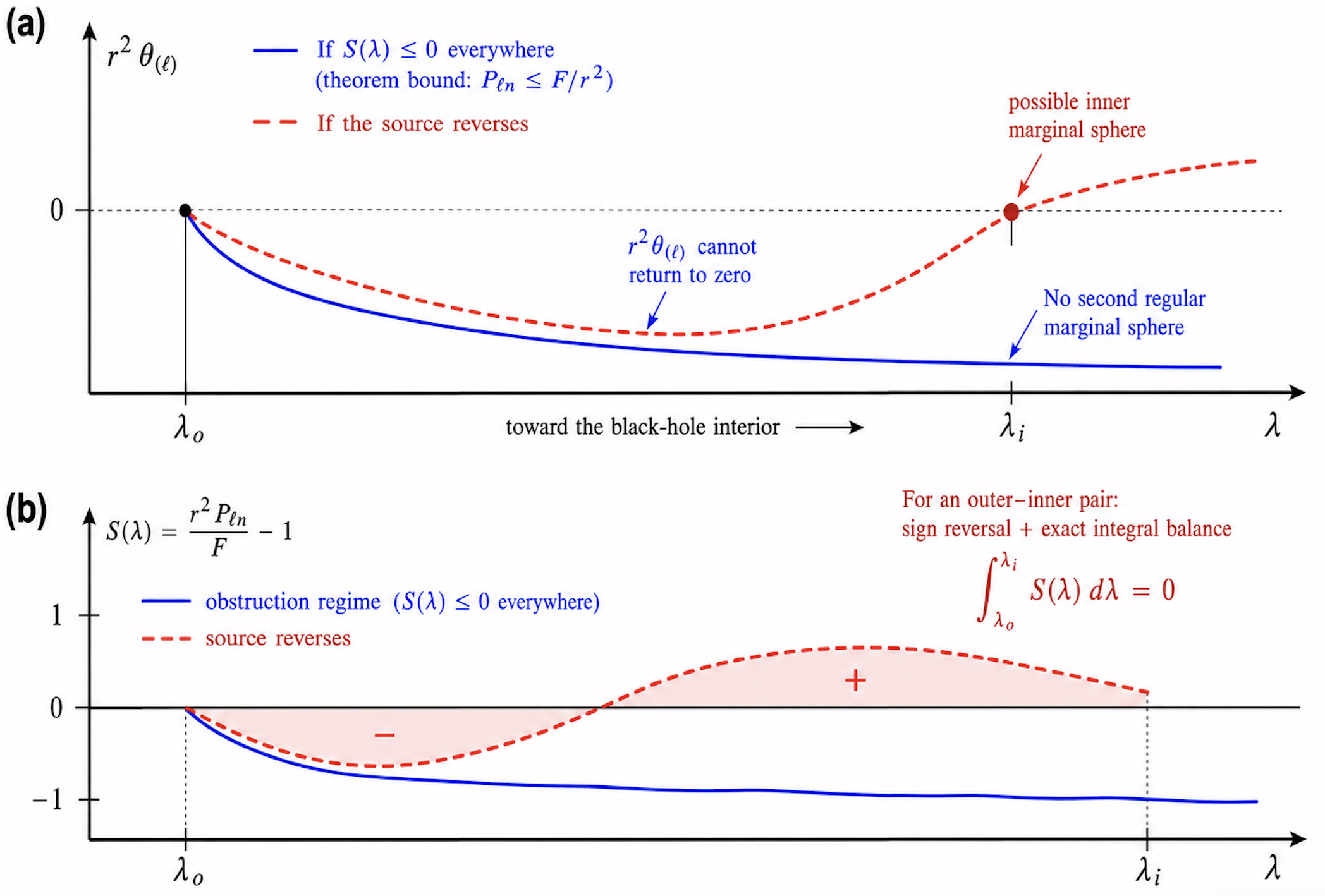}
	\caption{
		Schematic illustration of the obstruction theorem and the necessary
		source reversal for a regular outer--inner pair.
		Panel~(a) shows the evolution of the area-weighted outgoing
		expansion $r^2\theta_{(\ell)}$ along a future-directed ingoing
		null generator. Starting from the nondegenerate future outer
		marginal sphere at $\lambda_o$, the bound
		$\Pln\leq F/r^2$, or equivalently
		$\mathscr{S}\leq0$, makes $r^2\theta_{(\ell)}$
		nonincreasing once it becomes negative, thereby preventing its
		return to zero and excluding a second regular marginal sphere on
		the same generator. If the mixed source instead reverses the
		evolution sufficiently, a return to zero at a possible inner
		marginal sphere $\lambda_i$ becomes possible.
		Panel~(b) illustrates the corresponding source diagnostic
		$\mathscr{S}=r^2\Pln/F-1$. A regular outer--inner pair requires
		$\mathscr{S}<0$ at the outer marginal sphere and
		$\mathscr{S}>0$ at the inner marginal sphere, together with the
		exact integral balance
		$\int_{\lambda_o}^{\lambda_i}\mathscr{S}\,d\lambda=0$.
		The curves are schematic and are intended to illustrate the sign
		and monotonicity structure rather than a particular solution.}
	\label{fig:obstruction-schematic}
\end{figure*}

The inequality in Proposition~\ref{prop:reversal} may be supported by
a sufficiently positive mixed matter projection $T_{\ell n}$, by
the combined scalaron derivative terms, by the algebraic curvature
potential, or by a combination of these contributions. These channels
are not generally independent, since the scalaron and matter terms
are coupled through the trace equation.

A zero or sign change of $F$ is conceptually different: it does not
provide another regular contribution satisfying the inequality, but
instead invalidates the positive-coupling hypothesis and the divided
cross-focusing equation used in the theorem. Further possibilities
outside the theorem are that the candidate Cauchy boundary is
singular, is not marginal, or is not reached within a regular
double-null extension.

\begin{remark}[Degenerate case]
	If
	\begin{equation}
		\Pln=\frac{F}{r^2}
	\end{equation}
	at a marginal sphere, then
	$\Ccross_H=0$, and the sphere is degenerate rather than strictly
	outer or inner. Such an isolated equality does not alter
	Theorem~\ref{thm:no-inner}. Once
	$r^2\theta_{(\ell)}$ has become negative immediately inside a
	nondegenerate future outer marginal sphere, the bound
	$\Pln\leq F/r^2$ prevents it from increasing back to zero, even
	if the bound is saturated at later points.
\end{remark}

\section{Static limit and relation to radial criteria}
\label{sec:static-limit}

Consider a static, spherically symmetric metric with signature
$(- + + +)$,
\begin{equation}
	ds^2=-A(r)\,dt^2+\frac{dr^2}{B(r)}+r^2d\Omega^2.
	\label{eq:static-metric}
\end{equation}
Within a static region with $A>0$ and $B>0$, introduce the
ingoing Eddington--Finkelstein coordinate
\begin{equation}
	dv=dt+\frac{dr}{\sqrt{A(r)B(r)}}
	\label{eq:EF-coordinate}
\end{equation}
and define
\begin{equation}
	\Xi(r)\equiv\sqrt{\frac{A(r)}{B(r)}}.
	\label{eq:C-def}
\end{equation}
The resulting Eddington--Finkelstein form, which can then be continued
through a common-zero horizon under the regularity assumptions stated
below, is
\begin{equation}
	ds^2=-A(r)\,dv^2+2\Xi(r)\,dv\,dr+r^2d\Omega^2.
	\label{eq:EF-metric}
\end{equation}
We call a common-zero horizon $r=r_h$ regular in this chart when
$A(r_h)=B(r_h)=0$ and $\Xi$ admits a finite, strictly positive,
continuously differentiable extension to $r_h$. This assumption is
stronger and more useful than a ratio of first derivatives: it covers
both simple and higher-order common zeros without invoking an
indeterminate expression of the form $0/0$.

An adapted normalized null pair is
\begin{align}
	\ell^\mu&=(\partial_v)^\mu
	+\frac{A}{2\Xi}(\partial_r)^\mu,
	\label{eq:EF-l}\\
	n^\mu&=-\frac{1}{\Xi}(\partial_r)^\mu.
	\label{eq:EF-n}
\end{align}
It satisfies
\begin{align}
	\ell^2=n^2&=0,
	&\ell\cdot n&=-1,
	\notag\\
	\nabla_n n^\mu&=0,
	&\nabla_n\ell^\mu&=0.
	\label{eq:EF-normalization}
\end{align}
The expansions are
\begin{align}
	\theta_{(\ell)}
	&=\frac{A}{\Xi r}=\frac{\Xi B}{r},
	\label{eq:EF-theta-l}\\
	\theta_{(n)}&=-\frac{2}{\Xi r}.
	\label{eq:EF-theta-n}
\end{align}
The second form of $\theta_{(\ell)}$, obtained from
$A=\Xi^2B$, is the appropriate one for taking the regular horizon
limit. Acting with $n^\mu$ gives
\begin{equation}
	\Lie_n\theta_{(\ell)}
	=-\frac{1}{\Xi}\frac{d}{dr}
	\left(\frac{\Xi B}{r}\right).
	\label{eq:EF-C-general}
\end{equation}
At any regular common-zero horizon in the above sense, all terms
proportional to $B$ vanish, and one obtains the exact identity
\begin{equation}
	\Ccross_h
	\equiv\left.\Lie_n\theta_{(\ell)}\right|_{r_h}
	=-\frac{B'(r_h)}{r_h}.
	\label{eq:static-C}
\end{equation}
No simple-zero assumption was used. Consequently,
\begin{align}
	B'(r_h)>0&\Longleftrightarrow\text{future outer},
	\label{eq:Bprime-outer}\\
	B'(r_h)<0&\Longleftrightarrow\text{future inner},
	\label{eq:Bprime-inner}\\
	B'(r_h)=0&\Longleftrightarrow\text{degenerate},
	\label{eq:Bprime-degenerate}
\end{align}
provided the regular extension of $\Xi$ and the finiteness
assumptions of the cross-focusing equation hold. For a simple common
zero, the familiar relation
$\Xi_h^2=A'(r_h)/B'(r_h)$ follows as a special case; it is not used
for a degenerate horizon.

At the horizon, the area-weighted identity gives
\begin{equation}
	-1+\frac{r_h^2\left.\Pln\right|_{r_h}}{F_h}
	=-r_h B'(r_h).
	\label{eq:static-source-horizon}
\end{equation}
Thus the source equality $\Pln=F_h/r_h^2$ corresponds precisely to
the degenerate static case under the stated regularity assumptions,
while the strict outer and inner inequalities correspond to
$B'(r_h)>0$ and $B'(r_h)<0$, respectively.

For comparison, on an open static interval with $A>0$, $B>0$, and
$F>0$, define
\begin{equation}
	K^\mu=\left(\frac{1}{\sqrt A},\sqrt B,0,0\right),
	\qquad
	T_{KK}=T_{\mu\nu}K^\mu K^\nu.
	\label{eq:static-K}
\end{equation}
The same-direction radial projection yields the companion static law
\begin{equation}
	\frac{d}{dr}\left(\frac{B}{A}\right)
	=-\frac{r}{AF}
	\left(8\pi T_{KK}
	+K^\mu K^\nu\nabla_\mu\nabla_\nu F\right).
	\label{eq:static-ratio}
\end{equation}
This identity is restricted to connected static intervals and its
integrated form must retain the finite endpoint values of $B/A$.
By contrast, Eq.~\eqref{eq:area-weighted-cross-focusing} is regular
across the outer horizon and controls the evolution through a
nonstatic interior. The two laws are therefore complementary rather
than interchangeable.

{The relation between the static and double-null
	formulations is therefore complementary rather than a direct
	one-to-one dictionary. In the static analysis of
	Ref.~\cite{Munoz-Palma:2026dbd}, the ratio
	$Q\equiv B/A$ is governed, on each connected static interval, by the
	same-direction null projection
	$8\pi T_{KK}+K^\mu K^\nu\nabla_\mu\nabla_\nu F$.
	The present construction instead follows the transverse evolution of
	the outgoing expansion through the mixed source $\Pln$, and remains
	regular across a common-zero horizon and through the intervening
	nonstatic region. Likewise, the finite static endpoint datum
	$q_i=\lim_{r\to r_i}B/A$, which for simple common zeros reduces to
	$B'(r_i)/A'(r_i)$, should not be identified with the cross-focusing
	quantity $\Ccross_h=-B'(r_h)/r_h$: the former characterizes the
	relative normalization of the two static metric functions at the
	horizon, whereas the sign of the latter distinguishes outer, inner,
	and degenerate marginal horizons. The two approaches therefore probe
	different null projections of the same field equations and supply
	complementary information. They nevertheless share an important
	domain-of-validity requirement: the divided focusing relations require
	$F>0$. Accordingly, both analyses identify the zero of $F$ at
	$r=3M$ in the Multam\"aki--Vilja solution of
	Sec.~\ref{sec:MV-example} as an obstruction to extending the
	corresponding divided sign argument across the complete two-horizon
	region.}

\section{Consistency checks and illustrative solutions}
\label{sec:checks}

\subsection{General Relativity with a cosmological constant}

For
\begin{equation}
	f(R)=R-2\Lambda,
	\qquad
	F=1,
\end{equation}
the scalaron derivative terms vanish and
\begin{equation}
	FR-f=2\Lambda.
\end{equation}
Equation~\eqref{eq:P-invariant} therefore reduces to
\begin{equation}
	\Pln=8\pi T_{\ell n}+\Lambda.
	\label{eq:P-GR}
\end{equation}
At a future marginal sphere of the outgoing family,
Eq.~\eqref{eq:C-horizon} becomes
\begin{equation}
	\Ccross_H
	=-\frac{1}{r_H^2}
	+8\pi T_{\ell n}+\Lambda.
	\label{eq:C-GR}
\end{equation}

Equivalently, the area-weighted cross-focusing identity reads
\begin{equation}
	\Lie_n\!\left(r^2\theta_{(\ell)}\right)
	=-1+r^2\left(8\pi T_{\ell n}+\Lambda\right).
	\label{eq:area-weighted-GR}
\end{equation}
Thus, in General Relativity, the local outer/inner classification is
controlled by the competition between the spherical-curvature term
$1/r^2$, the mixed matter projection, and the cosmological constant.

{The recent no-inner-horizon theorem of
	Peng and An for static anisotropic black holes in Einstein
	gravity~\cite{Peng:2026baq} is based on a different mechanism from
	the one considered here. Their obstruction is formulated in terms of
	anisotropy among the transverse directions, encoded through
	differences of effective pressures and an associated radial integral
	condition inside the event horizon. The present GR limit,
	Eq.~\eqref{eq:C-GR}, instead follows from spherical cross-focusing
	and controls the transverse evolution of the outgoing null expansion
	through the mixed projection $T_{\ell n}$. Thus, the two results
	probe distinct structures of the Einstein equations: the
	Peng--An theorem exploits anisotropic transverse stresses in a
	static interior, whereas the present criterion is a horizon-regular
	mixed-null propagation statement applicable to spherical geometry.
	They should therefore be regarded as complementary no-inner-horizon
	diagnostics rather than as different formulations of the same
	condition.}

For Schwarzschild spacetime,
\begin{equation}
	T_{\ell n}=0,
	\qquad
	\Lambda=0,
\end{equation}
and hence
\begin{equation}
	\Ccross_H=-\frac{1}{r_H^2}<0.
\end{equation}
The Schwarzschild event horizon is therefore a nondegenerate future
outer marginal horizon in the orientation adopted here.

For Schwarzschild--de Sitter spacetime, consider the nondegenerate
parameter range
\begin{equation}
	0<9M^2\Lambda<1.
	\label{eq:SdS-nondegenerate-range}
\end{equation}
The static metric function is
\begin{equation}
	A(r)=B(r)
	=1-\frac{2M}{r}-\frac{\Lambda r^2}{3}.
	\label{eq:SdS-metric-function}
\end{equation}
The spacetime then possesses distinct black-hole and cosmological
horizons, denoted by $r_b$ and $r_c$, respectively. Since the
geometry is vacuum apart from the cosmological constant,
\begin{equation}
	\Ccross_H
	=\Lambda-\frac{1}{r_H^2}.
	\label{eq:C-SdS}
\end{equation}

The degenerate Nariai limit occurs at
\begin{equation}
	r_N=\frac{1}{\sqrt{\Lambda}},
	\qquad
	9M^2\Lambda=1.
\end{equation}
In the nondegenerate range,
\begin{equation}
	r_b<\frac{1}{\sqrt{\Lambda}}<r_c.
\end{equation}
Consequently,
\begin{align}
	\Ccross_b
	&=\Lambda-\frac{1}{r_b^2}<0,
	\notag\\
	\Ccross_c
	&=\Lambda-\frac{1}{r_c^2}>0.
\end{align}
Accordingly, the black-hole horizon is future outer, whereas the
cosmological horizon has the future-inner cross-focusing sign with
respect to the normalized null frame and orientation adopted here.

This classification agrees with the static horizon identity
$\Ccross_h=-B'(r_h)/r_h$. Indeed, using $B(r_h)=0$, one obtains
\begin{equation}
	B'(r_h)
	=\frac{1-\Lambda r_h^2}{r_h}.
\end{equation}
It follows that $B'(r_b)>0$ and $B'(r_c)<0$, reproducing the
respective outer and inner cross-focusing signs obtained from
Eq.~\eqref{eq:C-SdS}.

The cosmological horizon is not an inner black-hole horizon or a
Cauchy horizon. This example therefore does not test
Theorem~\ref{thm:no-inner}; rather, it verifies the local
cross-focusing classification and illustrates that the term
``future inner'' refers to the trapping-horizon sign and does not, by
itself, determine the global causal role of the marginal surface.

\subsection{Reissner--Nordstr\"om--(anti-)de Sitter}

For a radial Maxwell field,
\begin{equation}
	8\pi T_{\ell n}=\frac{Q^2}{r^4},
	\qquad
	T_{kk}=0
	\label{eq:Maxwell-projections}
\end{equation}
for either principal radial null direction $k^\mu$. Therefore,
\begin{equation}
	\Ccross_H=-\frac{1}{r_H^2}
	+\frac{Q^2}{r_H^4}+\Lambda.
	\label{eq:C-RNdS}
\end{equation}
This illustrates the distinction between same-direction and
cross-focusing: the radial Maxwell field saturates the same-direction
null projection but contributes positively to the mixed source.

Equivalently, the area-weighted cross-focusing identity becomes
\begin{equation}
	\Lie_n\!\left(r^2\theta_{(\ell)}\right)
	=-1+\frac{Q^2}{r^2}+\Lambda r^2.
	\label{eq:area-weighted-RNdS}
\end{equation}
At a marginal sphere, the outer, inner, and degenerate cases are
therefore distinguished according to whether
\begin{equation}
	\frac{Q^2}{r_H^2}+\Lambda r_H^2
\end{equation}
is smaller than, greater than, or equal to unity, respectively.

For asymptotically flat Reissner--Nordstr\"om, the two nonextremal
roots obey
\begin{equation}
	r_+r_-=Q^2,
	\qquad
	r_+>r_-.
\end{equation}
Equation~\eqref{eq:C-RNdS} gives
\begin{align}
	\Ccross_+
	&=\frac{r_--r_+}{r_+^3}<0,
	\label{eq:C-plus}\\
	\Ccross_-
	&=\frac{r_+-r_-}{r_-^3}>0.
	\label{eq:C-minus}
\end{align}
The outer event horizon and inner Cauchy horizon are therefore
classified correctly.

In particular, the Maxwell mixed stress satisfies the necessary
inner-horizon threshold at $r_-$:
\begin{equation}
	\frac{Q^2}{r_-^4}-\frac{1}{r_-^2}
	=\frac{r_+-r_-}{r_-^3}>0.
	\label{eq:RN-inner-threshold}
\end{equation}
Thus, the existence of the Reissner--Nordstr\"om inner horizon is
consistent with the source-reversal condition of
Proposition~\ref{prop:reversal}, while the sufficient bound of
Theorem~\ref{thm:no-inner} is necessarily violated at that horizon.

The same classification follows from the static identity
$\Ccross_h=-B'(r_h)/r_h$, with
\begin{equation}
	B(r)=1-\frac{2M}{r}+\frac{Q^2}{r^2}.
\end{equation}
Using $Q^2=r_+r_-$, one obtains
\begin{equation}
	B'(r_+)=\frac{r_+-r_-}{r_+^2}>0,
	\qquad
	B'(r_-)=-\frac{r_+-r_-}{r_-^2}<0.
\end{equation}
These relations reproduce Eqs.~\eqref{eq:C-plus}
and~\eqref{eq:C-minus}.

In the extremal limit,
\begin{equation}
	r_+=r_-=|Q|,
\end{equation}
one has
\begin{equation}
	\Ccross_H=0.
\end{equation}
The coincident horizon is therefore degenerate rather than strictly
outer or inner and lies outside the nondegenerate case addressed by
Theorem~\ref{thm:no-inner}.

For nonzero $\Lambda$, Eq.~\eqref{eq:C-RNdS} remains the appropriate
local classification formula at every regular nondegenerate root of
the Reissner--Nordstr\"om--(anti-)de Sitter metric function. The
global causal interpretation of each root must nevertheless be
established separately: the cross-focusing sign classifies a marginal
surface as outer or inner, but does not by itself determine whether
that surface is an event, cosmological, or Cauchy horizon.

\subsection{A scalaron-active charged inner-horizon benchmark}
\label{sec:Tang-benchmark}

{We emphasize before proceeding that
Theorem~\ref{thm:no-inner} and Proposition~\ref{prop:reversal} do not
depend on the specific geometry considered below: the theorem is a
statement about any regular null generator on which the source bound
holds, established independently of any particular solution. The
example that follows should therefore be read as a consistency check
verifying that the framework operates correctly on a genuine
nonconstant-curvature $f(R)$ solution, not as evidence in itself for
the theorem.}

A genuinely nonconstant-$F$ application is furnished by the exact
four-dimensional Maxwell--$f(R)$ solution of Tang, Wang, and
Papantonopoulos~\cite{Tang:2019qiy}. For $c_1=0$, $k=1$, and vanishing
cosmological constant, the relations $R(r)=1/r^2$ and
$f_{\rm mod}(r)=2c_2/r$ displayed in Ref.~\cite{Tang:2019qiy} give the
total Lagrangian
\begin{equation}
	f(R)=R+2c_2\sqrt{R},
	\qquad c_2<0,
	\label{eq:Tang-action}
\end{equation}
and hence $F=f_R=1+c_2r$. The corresponding spherical branch has
\begin{align}
	A(r)=B(r)&=\frac12+\frac{1}{3c_2r}
	+\frac{q^2}{4r^2},
	\label{eq:Tang-metric}\\
	R(r)&=\frac{1}{r^2},
	\qquad
	F(r)=1+c_2r.
	\label{eq:Tang-F}
\end{align}
Writing
\begin{equation}
	m\equiv-\frac{1}{3c_2}>0,
\end{equation}
the metric function becomes
\begin{equation}
	B(r)=\frac12-\frac{m}{r}+\frac{q^2}{4r^2}.
	\label{eq:Tang-metric-m}
\end{equation}
For
\begin{equation}
	0<q^2<2m^2,
	\label{eq:Tang-nonextremal}
\end{equation}
there are two distinct positive roots
\begin{equation}
	r_\pm=m\pm\sqrt{m^2-\frac{q^2}{2}},
	\qquad
	r_+r_-=\frac{q^2}{2}.
	\label{eq:Tang-roots}
\end{equation}
The region $r_-<r<r_+$ is future trapped in the ingoing
Eddington--Finkelstein orientation used above. Moreover,
\begin{equation}
	F(r)=1-\frac{r}{3m}>\frac13
	\qquad (r_-\leq r\leq r_+),
	\label{eq:Tang-F-positive}
\end{equation}
because $r_+<2m$. Thus the scalaron is nonconstant, strictly
positive, and regular throughout the complete inter-horizon segment.
The solution is used here as an exact local benchmark: $F$ vanishes
at $r=3m>r_+$, so this branch does not satisfy the positive-coupling
condition on its entire asymptotic exterior.

Here $\Xi=1$, so that
$n=-\partial_r$ and
$\ell=\partial_v+(B/2)\partial_r$. Rather than inferring the mixed
source from the geometric identity, it can be evaluated independently
term by term. With the Maxwell normalization of
Ref.~\cite{Tang:2019qiy}, the radial electric field gives
$8\pi T_{\ell n}=q^2/(4r^4)$. Since
$F=1-r/(3m)$ is linear in $r$, direct evaluation of the Hessian,
the d'Alembertian, and the curvature potential yields
\begin{align}
	8\pi T_{\ell n}
	&=\frac{q^2}{4r^4},
	\notag\\
	\nabla_\ell\nabla_nF
	&=\frac{B'(r)}{6m},
	\notag\\
	\Box F
	&=-\frac{1}{3m}
	\left(B'(r)+\frac{2B(r)}{r}\right),
	\notag\\
	\frac12(FR-f)
	&=\frac{1}{6mr}.
	\label{eq:Tang-source-pieces}
\end{align}
Substituting Eq.~\eqref{eq:Tang-metric-m} into these four contributions
gives
\begin{align}
	\Pln
	&=\frac{q^2}{4r^4}+\frac{B'}{6m}
	-\frac{1}{3m}\left(B'+\frac{2B}{r}\right)
	+\frac{1}{6mr}
	\notag\\
	&=F(r)\left(
	\frac{1}{2r^2}+\frac{q^2}{4r^4}
	\right).
	\label{eq:Tang-source-direct}
\end{align}
Hence the total source satisfies
\begin{equation}
	\frac{\Pln}{F}
	=\frac{1}{2r^2}+\frac{q^2}{4r^4},
	\label{eq:Tang-total-source}
\end{equation}
and the diagnostic defined in
Eq.~\eqref{eq:source-reversal-function} becomes
\begin{equation}
	\mathscr{S}(r)
	=-\frac12+\frac{q^2}{4r^2}.
	\label{eq:Tang-source-diagnostic}
\end{equation}
This source-side calculation is independent of the area-weighted
cross-focusing identity. As a geometric cross-check,
$r^2\theta_{(\ell)}=rB$ and $n=-\partial_r$ give
$\Lie_n(r^2\theta_{(\ell)})=-d(rB)/dr$, which reproduces exactly
Eq.~\eqref{eq:Tang-source-diagnostic}.

At the two horizons,
\begin{align}
	\mathscr{S}(r_+)
	&=\frac{r_--r_+}{2r_+}<0,
	\label{eq:Tang-S-plus}\\
	\mathscr{S}(r_-)
	&=\frac{r_+-r_-}{2r_-}>0.
	\label{eq:Tang-S-minus}
\end{align}
The outer horizon is therefore future outer and the inner horizon
future inner. This exact scalaron-active geometry realizes the source
reversal required by Proposition~\ref{prop:reversal}; correspondingly,
the sufficient bound of Theorem~\ref{thm:no-inner} fails before the
inner horizon is reached.

The integral balance is also satisfied explicitly. Since the affine
parameter increases inward and $d\lambda=-dr$,
\begin{align}
	\int_{\lambda_+}^{\lambda_-}\mathscr{S}\,d\lambda
	&=\int_{r_-}^{r_+}
	\left(-\frac12+\frac{q^2}{4r^2}\right)dr
	\notag\\
	&=0,
	\label{eq:Tang-integral-balance}
\end{align}
where Eq.~\eqref{eq:Tang-roots} was used. The example therefore checks
all three levels of the framework in a solution with dynamic curvature:
the local outer/inner signs, the necessary source reversal, and the
exact integrated compensation. It also shows why the theorem is a
conditional obstruction rather than a generic prohibition of charged
inner horizons in metric $f(R)$ gravity.

These properties are displayed explicitly in
Fig.~\ref{fig:Tang-scalaron-active}. Introducing the dimensionless
radius $x\equiv r/m$ and charge parameter
$\hat q\equiv q/m$, the figure uses the representative
nonextremal choice $\hat q=1$. It simultaneously shows the two
zeros of the metric function, the strict positivity of the scalaron
throughout the inter-horizon region, and the zero crossing of
$\mathscr{S}$ required by the source-reversal condition.

\begin{figure}[th!]
	\centering
	\includegraphics[width=\columnwidth]
	{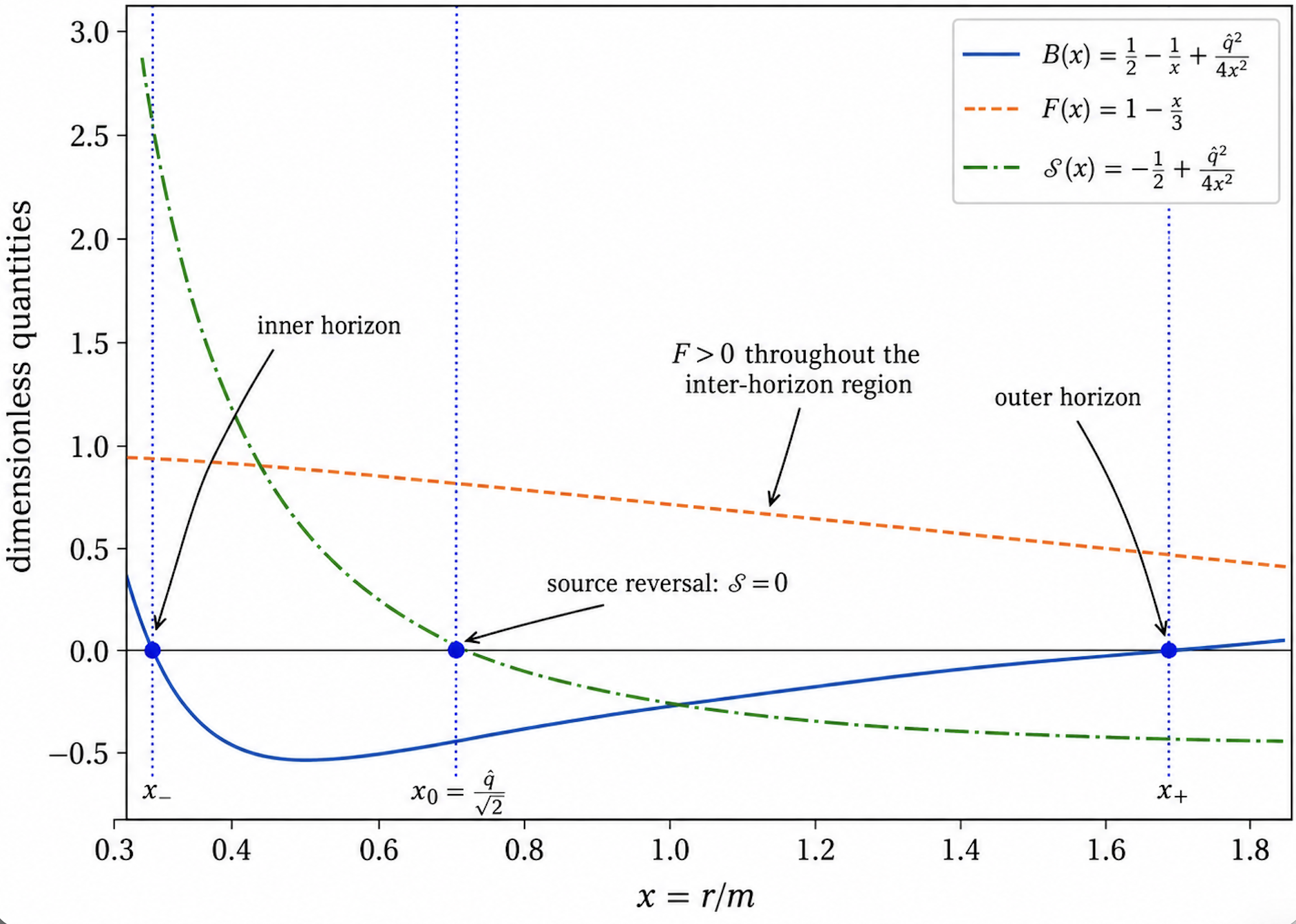}
	\caption{
		Exact scalaron-active realization of the inner-horizon mechanism
		for the charged $f(R)$ solution of
		Ref.~\cite{Tang:2019qiy}. We introduce
		$x\equiv r/m$ and $\hat q\equiv q/m$, and choose the
		representative nonextremal value $\hat q=1$.
		The solid curve shows
		$B(x)=1/2-1/x+\hat q^{\,2}/(4x^2)$, whose zeros
		$x_-$ and $x_+$ locate the inner and outer horizons,
		respectively. The dashed curve shows the scalaron
		$F(x)=1-x/3$, which remains strictly positive throughout
		the complete inter-horizon interval. The dash-dotted curve is
		the source diagnostic
		$\mathscr{S}(x)=-1/2+\hat q^{\,2}/(4x^2)$.
		It vanishes at
		$x_0=\hat q/\sqrt{2}$, with
		$x_-<x_0<x_+$, and changes from
		$\mathscr{S}>0$ near the inner horizon to
		$\mathscr{S}<0$ near the outer horizon.
		The figure therefore displays directly the scalaron positivity,
		the source reversal, and the two-horizon structure entering the
		exact benchmark.}
	\label{fig:Tang-scalaron-active}
\end{figure}

\subsection{\texorpdfstring{Constant-curvature $f(R)$ solutions}
	{Constant-curvature f(R) solutions}}

Let $R=R_0$ and $F=F_0>0$ be constant. Define
\begin{equation}
	\Lambda_{\rm eff}
	\equiv\frac{F_0R_0-f(R_0)}{2F_0}.
	\label{eq:Lambda-eff}
\end{equation}
Then
\begin{equation}
	\frac{\Pln}{F_0}
	=\frac{8\pi}{F_0}T_{\ell n}
	+\Lambda_{\rm eff},
	\label{eq:P-constant-curvature}
\end{equation}
and
\begin{equation}
	\Ccross_H=-\frac{1}{r_H^2}
	+\frac{8\pi}{F_0}T_{\ell n}
	+\Lambda_{\rm eff}.
	\label{eq:C-constant-curvature}
\end{equation}
Thus the criterion reduces to the General-Relativity expression with
a rescaled matter coupling and an effective cosmological constant.
Constant curvature does not activate the derivative part of the
scalaron mechanism.

Equivalently, the area-weighted cross-focusing identity becomes
\begin{equation}
	\Lie_n\!\left(r^2\theta_{(\ell)}\right)
	=-1+r^2\left(
	\frac{8\pi}{F_0}T_{\ell n}
	+\Lambda_{\rm eff}
	\right).
	\label{eq:area-weighted-constant-curvature}
\end{equation}
Hence a regular future inner marginal sphere must satisfy
\begin{equation}
	\frac{8\pi}{F_0}T_{\ell n}
	+\Lambda_{\rm eff}
	>\frac{1}{r_H^2},
	\label{eq:constant-curvature-inner-threshold}
\end{equation}
whereas the opposite strict inequality characterizes a future outer
marginal sphere.

The constant-curvature assumption must also be compatible with the
trace equation. Since $\Box F=0$, Eq.~\eqref{eq:trace} reduces to
\begin{equation}
	F_0R_0-2f(R_0)=8\pi T.
	\label{eq:constant-curvature-trace}
\end{equation}
Thus, if $R_0$ is constant, the matter trace must also be constant on
the region considered. In vacuum, or for traceless matter, one has
\begin{equation}
	F_0R_0-2f(R_0)=0,
\end{equation}
so that
\begin{equation}
	\Lambda_{\rm eff}=\frac{R_0}{4}.
	\label{eq:Lambda-eff-vacuum}
\end{equation}
Accordingly, constant curvature removes the derivative scalaron terms
but retains the algebraic curvature-potential contribution encoded in
$\Lambda_{\rm eff}$.

\subsection{\texorpdfstring{A nonconstant-curvature two-horizon solution and the positivity of $F$}{A nonconstant-curvature two-horizon solution and the positivity of F}}
\label{sec:MV-example}

Multam\"aki and Vilja found a vacuum solution with
\begin{align}
	A(r)&=1-\frac{2M}{r}-\frac{\Lambda r^2}{3},
	\label{eq:MV-A}\\
	B(r)&=\frac{A(r)}{2},
	\label{eq:MV-B}\\
	F(r)&=1-\frac{r}{3M}.
	\label{eq:MV-F}
\end{align}
For
\begin{equation}
	0<9M^2\Lambda<1,
	\label{eq:SdS-range}
\end{equation}
the redshift function has a black-hole root and a cosmological root,
with a static untrapped region between them. This pair is therefore
not an outer-event/inner-Cauchy-horizon configuration and is not an
application of Theorem~\ref{thm:no-inner}.

More precisely, the intended black-hole outer--inner application starts
from a future outer marginal surface and follows ingoing null generators
toward a possible future inner marginal surface. Such applications
normally traverse a future trapped region, although
Theorem~\ref{thm:no-inner} itself does not require
$\theta_{(n)}<0$ at every intermediate point. In the present solution,
the region between the black-hole and cosmological roots satisfies
$A>0$ and $B>0$ and is therefore static and untrapped. The second
root is a cosmological horizon, not an inner black-hole or Cauchy
horizon.

Figure~\ref{fig:MV-viability} summarizes the structure of this solution.
It shows the metric function $A(r)$ (equivalently $2B(r)$) together
with the scalaron $F(r)$, highlighting the two horizon radii and the
intermediate point $r=3M$ where $F$ changes sign.

\begin{figure}[th!]
	\centering
	\includegraphics[width=\columnwidth]{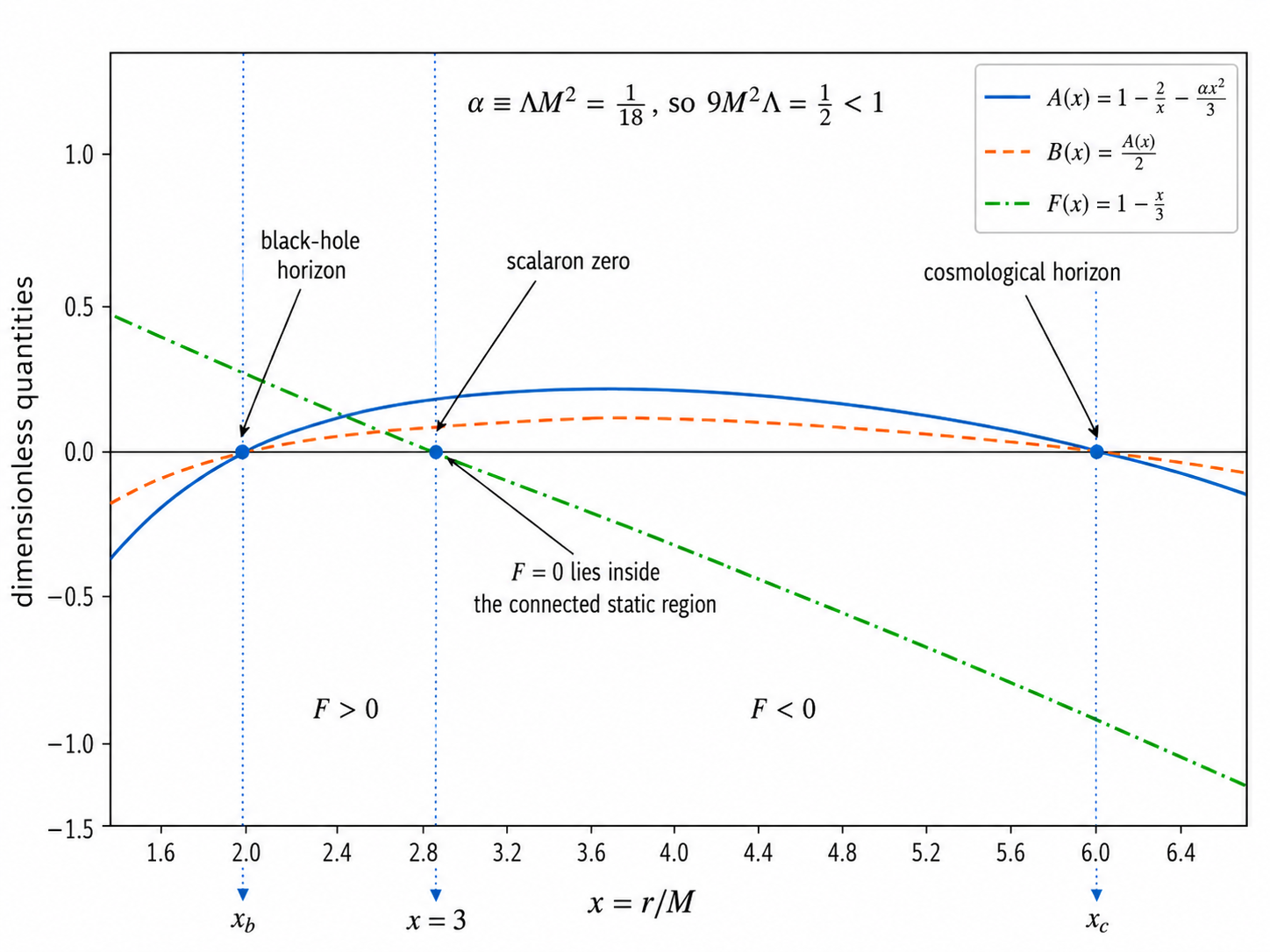}
	\caption{Scalaron-viability limitation in the two-horizon
		Multam\"aki--Vilja solution. The metric function $A(r)$
		(equivalently $2B(r)$) has two positive zeros, corresponding to a
		black-hole horizon and a cosmological horizon, while the scalaron
		$F(r)=1-r/(3M)$ vanishes at $r=3M$, which lies strictly between
		the two horizons in the nondegenerate range $0<9M^2\Lambda<1$.
		Thus the static interval connecting the two horizons contains a zero
		of $F$, so the divided Einstein-form equations and the source ratio
		$\mathcal{P}_{\ell n}/F$ are not defined throughout the whole
		region. The figure illustrates that this example is not a
		black-hole outer/inner-horizon configuration and instead serves as a
		global viability check for the positivity condition $F>0$.
	}
	\label{fig:MV-viability}
\end{figure}

Nevertheless, the solution provides a useful viability check. One has
\begin{equation}
	A(3M)=\frac{1-9M^2\Lambda}{3}>0,
	\label{eq:A-3M}
\end{equation}
so $r=3M$ lies between the two roots, while
\begin{equation}
	F(3M)=0.
	\label{eq:F-zero}
\end{equation}

Since $B(3M)=A(3M)/2>0$, the zero of $F$ lies strictly inside the
connected static region. Moreover,
\begin{equation}
	F(r)>0 \quad \text{for} \quad r<3M,
	\qquad
	F(r)<0 \quad \text{for} \quad r>3M.
	\label{eq:F-sign-MV}
\end{equation}
Thus no connected interval spanning both horizons satisfies the
positive-coupling hypothesis $F>0$.

As is clear from Fig.~\ref{fig:MV-viability}, the zero of the scalaron
occurs inside the same static interval that contains both horizon
radii. This is precisely what obstructs the use of the divided
cross-focusing equation on the full two-horizon region.

The Einstein-form equations obtained by dividing by $F$, and any
focusing inequality that contains $1/F$, are not defined through
that point. The original undivided field equations may
remain algebraically meaningful if the metric and scalaron are
regular, but the coefficient of the Ricci tensor vanishes at
$F=0$. Consequently, the effective gravitational coupling and the
divided source ratio $\Pln/F$ become degenerate there, and the
hypotheses of the cross-focusing theorem fail.

The correct lesson is therefore limited but important: scalaron
viability conditions must be checked globally on the region to which
a focusing argument is applied. The zero of $F$ is not
a mechanism by which this geometry evades
Theorem~\ref{thm:no-inner}; the black-hole--cosmological-horizon pair
lies outside the theorem's geometric setting from the outset. Rather,
the example shows that even an exact, apparently regular two-horizon
metric cannot be used in a divided $f(R)$ focusing argument unless
$F>0$ is verified throughout the complete domain under study.

\subsection{\texorpdfstring{Connection with dynamical $f(R)$ interiors}
	{Connection with dynamical f(R) interiors}}

Double-null simulations of charged collapse in $f(R)$ gravity have
found mass inflation and strong scalaron evolution near the
inner-horizon region~\cite{Hwang:2011kg,Guo:2015ira}. Other
collapse studies show that $F$ may approach zero or become unbounded
in high-curvature regimes, depending on the model and initial
data~\cite{Guo:2013dha,Guo:2015ira}.

{Since the physical motivation for this entire
	framework originates in the inner-horizon problem, it is worth
	stating plainly what the results established so far do and do not
	say about it. Theorem~\ref{thm:no-inner} and
	Proposition~\ref{prop:reversal} constrain the cross-focusing
	conditions under which a regular inner marginal horizon can exist at
	all; they do not determine whether, once present, such a horizon
	survives the nonlinear mass-inflation instability discussed in
	Sec.~\ref{sec:scope}. The present framework should therefore be
	understood as characterizing the geometric and source prerequisites
	for the existence of a regular inner marginal horizon, logically
	distinct from the subsequent question of its nonlinear stability.}

These results motivate evaluating the diagnostic of
Proposition~\ref{prop:reversal} directly in numerical data, but the
published simulations do not by themselves supply the complete
decomposition of $\Pln$ used here. Strong scalaron
evolution can modify the mixed derivative contribution to $\Pln$ and
thereby change the cross-focusing balance. A limit $F\to0$ invalidates
the divided equation and the positive-coupling hypothesis used in
Theorem~\ref{thm:no-inner}, whereas an unbounded $F$ does not by
itself violate $F>0$ but may signal a loss of regularity or render the
source ratio $\Pln/F$ singular or ill behaved.

The cited simulations do not, by themselves, establish whether the
bound $\Pln\leq F/r^2$ is satisfied, since this requires evaluating
the complete mixed source along the relevant ingoing null generators.
They instead motivate the use of the diagnostic quantity
\begin{equation}
	\mathscr{S}(\lambda)
	=\frac{r^2\Pln}{F}-1,
	\label{eq:dynamical-source-diagnostic}
\end{equation}
whose sign controls the area-weighted evolution of the outgoing
expansion through
Eq.~\eqref{eq:area-weighted-cross-focusing}. A direct application to
numerical collapse solutions would therefore require extracting
$r$, $F$, the matter projections, and the scalaron derivatives along
the same regular null segment.

The numerical literature also emphasizes that a Cauchy horizon
approached during mass inflation need not coincide with an inner
apparent horizon~\cite{Hwang:2011kg}. This
distinction is essential: a Cauchy horizon is defined globally as a
boundary of the domain of dependence, whereas an inner marginal or
trapping horizon is characterized quasi-locally by the null
expansions. This is precisely why the result proved here is
formulated first as a theorem about regular inner marginal horizons,
with the Cauchy-horizon statement given only under the additional
assumptions of Corollary~\ref{cor:Cauchy}.

\section{Scope, limitations, and physical interpretation}
\label{sec:scope}

Theorem~\ref{thm:no-inner} is a horizon-regular propagation
obstruction, but it is not an unconditional theorem about every
possible Cauchy boundary. Its scope is best summarized by separating
three distinct questions.

\paragraph{Existence of a regular inner marginal horizon}

Theorem~\ref{thm:no-inner} directly addresses the existence of a second
regular marginal sphere along a future-directed ingoing generator. If
the generator leaves a regular, nondegenerate future outer marginal
sphere, $F>0$, and Eq.~\eqref{eq:source-bound} holds throughout the
subsequent regular segment, then $r^2\theta_{(\ell)}$ remains strictly
negative and no later regular future marginal sphere of the outgoing
family can occur. The theorem does not require
$\theta_{(n)}<0$ at every intermediate point, although the intended
black-hole outer--inner application normally traverses a future trapped
region. Any regular nondegenerate future inner marginal horizon must
instead satisfy the strict local reversal condition
Eq.~\eqref{eq:necessary-inner}. Because the construction is entirely
double-null, these statements apply to both static and dynamical
spherically symmetric geometries.

\paragraph{Identification with a Cauchy horizon}

In Reissner--Nordstr\"om-type maximal analytic extensions, a static,
nondegenerate inner Killing horizon is marginal and also forms a Cauchy
horizon; under those additional global assumptions,
Corollary~\ref{cor:Cauchy} obstructs the corresponding regular inner
boundary. The identification is not automatic in a general spacetime:
a Cauchy horizon is a global boundary of a domain of dependence, whereas
an inner marginal or trapping horizon is defined quasi-locally by the
null expansions. Consequently, the absence of a regular future inner
marginal horizon implies the absence of a Cauchy horizon only when the
additional global hypotheses of the corollary are satisfied.

\paragraph{Stability and mass inflation}

The theorem concerns the existence of a regular marginal boundary and
the cross-focusing evolution leading to it; it does not address
nonlinear stability. Even when a background admits an inner Cauchy
horizon, perturbations may produce mass inflation and turn the horizon
into a weak or strong null
singularity~\cite{Poisson:1990eh,Ori:1991zz,Hwang:2011kg}.
Conversely, bounded scalar invariants do not by themselves establish
extendibility of the metric or determinism beyond the boundary. A
complete strong-cosmic-censorship analysis requires control of
perturbations, the differentiability of the limiting geometry, and the
structure of the maximal Cauchy development.

Within these limits, the source threshold has a direct interpretation:
$F/r^2$ is the scalaron-weighted intrinsic-curvature scale of the
symmetry spheres, and a regular nondegenerate future inner marginal
horizon requires $\Pln$ to exceed it locally. If an outer--inner pair
is connected by a regular generator, the diagnostic
$\mathscr{S}=r^2\Pln/F-1$ of
Eq.~\eqref{eq:source-reversal-function} must reverse sign and satisfy
the exact balance~\eqref{eq:source-balance}. In General Relativity,
electric charge can supply the required positive mixed stress; in
metric $f(R)$ gravity, matter, scalaron derivatives, and the algebraic
curvature potential may contribute jointly. A model-level
no-inner-marginal-horizon theorem therefore requires conditions strong
enough to maintain Eq.~\eqref{eq:source-bound} throughout the relevant
regular null segment. A zero or sign change of $F$ is different: it
invalidates the positive-coupling hypothesis and the divided
cross-focusing equation rather than providing an additional regular
source channel.

\section{Conclusions}
\label{sec:conclusions}

We have developed a horizon-regular double-null framework for analysing
regular inner marginal horizons in spherically symmetric metric $f(R)$
gravity. The key result is an exact area-weighted cross-focusing transport
law that remains regular across an outer horizon and through the
intervening nonstatic region, while keeping the physical matter,
scalaron-derivative, and algebraic $f(R)$ contributions explicitly
separated.

This leads to a finite-segment propagation obstruction. If the mixed
matter--scalaron source remains below the scalaron-weighted curvature
threshold along a regular future-directed ingoing generator issuing from a
nondegenerate future outer marginal sphere, the outgoing expansion cannot
recover to zero and no second regular marginal sphere of the same family
can occur. Conversely, a regular nondegenerate future inner marginal
sphere requires the source to exceed that threshold locally. A connected
outer--inner pair therefore necessarily entails a source reversal together
with an exact integrated compensation along the intervening segment. This
result does not require the entire intermediate region to remain future
trapped.

The static reduction provides an independent consistency check and shows
that the outer, inner, and degenerate cases can be classified in
horizon-regular Eddington--Finkelstein coordinates without assuming simple
zeros of the metric functions. The degenerate case is therefore obtained
from a well-defined regularity condition rather than from an indeterminate
ratio of vanishing derivatives.

The examples clarify both the geometrical and physical content of the
criterion. Schwarzschild, Schwarzschild--de Sitter, and
Reissner--Nordstr\"om reproduce the expected local classifications. The
charged case is especially instructive because the usual same-direction
null projection can saturate while the mixed stress still supplies the
cross-focusing contribution required to distinguish outer from inner
horizons. The exact scalaron-active $f(R)$ solution of
Sec.~\ref{sec:Tang-benchmark} provides a nontrivial modified-gravity
realization: the scalaron remains positive throughout the inter-horizon
region, while the mixed source undergoes the required reversal and
satisfies the integrated balance exactly. By contrast, the
Multam\"aki--Vilja solution illustrates that the condition $F>0$ must be
verified throughout the complete region on which the divided focusing
equation is applied.

{None of the examples above directly tests the
	sufficient direction of Theorem~\ref{thm:no-inner}, namely, the
	propagation obstruction itself, in a genuinely scalaron-active
	regime with nonconstant $F$. The GR-type examples
	(Schwarzschild, Reissner--Nordstr\"om, and Schwarzschild--de Sitter)
	have $F\equiv1$, while the two solutions with nonconstant $F$
	(Sec.~\ref{sec:Tang-benchmark} and Sec.~\ref{sec:MV-example}) either
	realize the necessary source-reversal behaviour of
	Proposition~\ref{prop:reversal} or fail the positive-coupling
	hypothesis somewhere in the relevant region. Identifying a genuine
	single-horizon $f(R)$ solution with nonconstant $F$ for which the
	bound $\Pln\leq F/r^2$ can be verified along the complete regular
	future-directed ingoing null segment issuing from the outer marginal
	horizon would provide a direct nontrivial benchmark of the sufficient
	obstruction and is left for future work.}

The broader significance of the framework is that it identifies the
specific local combination of matter and modified-gravity terms that must
change for a regular inner marginal boundary to form. It therefore provides
a direct bridge between quasi-local horizon geometry and the dynamical
content of the gravitational theory, while avoiding the ambiguity of
absorbing scalaron contributions into an effective matter source.

The scope of the result should nevertheless remain clear. A future inner
marginal horizon is a quasi-local object and need not coincide with a
Cauchy horizon. The corresponding Cauchy-horizon corollary applies only
when the candidate global boundary is also a regular nondegenerate future
inner marginal horizon reached through a regular double-null extension.
Extremal boundaries, nonmarginal Cauchy horizons, and mass-inflation
singularities require additional global and dynamical analysis.

Several extensions are natural. A particularly important next step is to
evaluate the mixed source directly in double-null simulations of viable
$f(R)$ collapse, thereby identifying whether matter, scalaron derivatives,
or the algebraic curvature potential drives the required reversal near a
candidate inner horizon. Such an analysis may lead to model-dependent
sufficient conditions excluding regular inner marginal horizons and could
also clarify how the present criterion behaves as an inner horizon evolves
toward a mass-inflation singularity. It would also be valuable to determine
whether analogous horizon-regular cross-focusing obstructions arise in
scalar--tensor, Horndeski, and more general higher-curvature theories.

The central lesson is that the existence of a regular inner marginal
horizon is not governed by a single effective energy condition, but by the
evolution of a definite mixed source along the interior null direction.
The horizon-regular formulation developed here turns this statement into a
quantitative propagation obstruction, a necessary source-reversal
condition, and an exact integrated balance, providing a systematic basis
for future analytical and numerical studies of inner horizons in modified
gravity. {Stated at its most general, the result identifies
	a geometric mechanism, the competition between the mixed
	matter--scalaron source and the scalaron-weighted spherical-curvature
	scale, that governs the propagation of regular inner marginal
	horizons in spherical metric $f(R)$ gravity. The mathematical form of
	this mechanism is independent of the detailed matter model and of the
	specific $f(R)$ Lagrangian, although whether the required source
	inequality is actually satisfied is necessarily model dependent.}

\acknowledgments{FSNL acknowledges funding from the Funda\c{c}\~{a}o para a Ci\^{e}ncia e a Tecnologia (FCT) grant UID/04434/2025, and acknowledges support from the FCT Scientific Employment Stimulus contract with reference CEECINST/00032/2018.}

\appendix

\section{Derivation of the double-null equations}
\label{app:derivation}

For the metric~\eqref{eq:double-null}, with
$\ell^\mu=\partial_v$ and
$n^\mu=e^{2\sigma}\partial_u$, the nonvanishing connection
coefficients in the two-dimensional $(u,v)$ sector that are needed
below are
\begin{equation}
	\Gamma^u{}_{uu}=-2\sigma_{,u},
	\qquad
	\Gamma^v{}_{vv}=-2\sigma_{,v}.
	\label{eq:Christoffel}
\end{equation}
The radial null Ricci components are
\begin{align}
	R_{uu}&=-\frac{2}{r}
	\left(r_{,uu}+2\sigma_{,u}r_{,u}\right),
	\label{eq:Ruu}\\
	R_{vv}&=-\frac{2}{r}
	\left(r_{,vv}+2\sigma_{,v}r_{,v}\right).
	\label{eq:Rvv}
\end{align}
The mixed Einstein component is
\begin{equation}
	G_{uv}=\frac{2r_{,uv}}{r}
	+\frac{2r_{,u}r_{,v}}{r^2}
	+\frac{e^{-2\sigma}}{r^2}.
	\label{eq:Guv}
\end{equation}

Since $G_{\ell n}=e^{2\sigma}G_{uv}$, while
\begin{equation}
	\Lie_n\theta_{(\ell)}
	=2e^{2\sigma}\left(
	\frac{r_{,uv}}{r}
	-\frac{r_{,u}r_{,v}}{r^2}
	\right),
	\label{eq:n-theta-coordinate}
\end{equation}
and
\begin{equation}
	\theta_{(\ell)}\theta_{(n)}
	=\frac{4e^{2\sigma}r_{,u}r_{,v}}{r^2},
\end{equation}
one obtains
\begin{equation}
	G_{\ell n}
	=\Lie_n\theta_{(\ell)}
	+\theta_{(\ell)}\theta_{(n)}
	+\frac{1}{r^2},
\end{equation}
which is Eq.~\eqref{eq:Gln-geometry}.

For a spherically symmetric scalar $F(u,v)$, direct evaluation of
the d'Alembertian gives
\begin{equation}
	\Box F=-2e^{2\sigma}
	\left[
	F_{,uv}
	+\frac{r_{,u}F_{,v}+r_{,v}F_{,u}}{r}
	\right].
	\label{eq:boxF}
\end{equation}
Furthermore, because the mixed connection coefficients
$\Gamma^\mu{}_{uv}$ vanish in the two-dimensional sector,
\begin{equation}
	\nabla_\ell\nabla_n F
	\equiv
	\ell^\mu n^\nu\nabla_\mu\nabla_\nu F
	=e^{2\sigma}F_{,uv}.
	\label{eq:Hln}
\end{equation}
Substitution of Eqs.~\eqref{eq:boxF} and~\eqref{eq:Hln} into
Eq.~\eqref{eq:P-invariant} yields Eq.~\eqref{eq:P-coordinate}.

The area-weighted cross-focusing identity follows from
\begin{equation}
	\Lie_n(r^2)=r^2\theta_{(n)}.
\end{equation}
Therefore,
\begin{align}
	\Lie_n\!\left(r^2\theta_{(\ell)}\right)
	&=r^2\Lie_n\theta_{(\ell)}
	+r^2\theta_{(\ell)}\theta_{(n)}
	\notag\\
	&=-1+\frac{r^2\mathcal{P}_{\ell n}}{F}.
\end{align}
This is the area-weighted identity used in the main obstruction
theorem.

For completeness, the affine property of $n^\mu$ follows from
\begin{align}
	n^u\partial_u n^u
	+\Gamma^u{}_{uu}(n^u)^2
	&=2\sigma_{,u}e^{4\sigma}
	-2\sigma_{,u}e^{4\sigma}
	\notag\\
	&=0.
\end{align}
Hence $\nabla_n n^\mu=0$. Since
$\ell^\mu=\partial_v$ has constant components and
$\Gamma^\mu{}_{uv}=0$ in the two-dimensional sector, one also has
\begin{equation}
	\nabla_n\ell^\mu=0.
\end{equation}
Thus the chosen null frame is parallel transported along the affinely
parametrized ingoing generators.

\end{document}